\documentclass[a4paper,twocolumn]{quantumarticle}
\pdfoutput=1
\usepackage[utf8]{inputenc}
\usepackage[english]{babel}
\usepackage[T1]{fontenc}
\usepackage{amsmath,amssymb,amsthm}
\usepackage{hyperref,braket}

\def\bbra#1{\mathinner{\langle\!\langle{#1}|}}
\def\kett#1{\mathinner{|{#1}\rangle\!\rangle}}
\def\bbrakett#1{\mathinner{\langle\!\langle{#1}\rangle\!\rangle}}

\newtheorem{dfn}{Definition}
\newtheorem{lmm}{Lemma}
\newtheorem{thm}{Theorem}
\newtheorem{cor}{Corollary}

\DeclareMathOperator{\Tr}{Tr}
\DeclareMathOperator{\openone}{\mathds{1}}

\DeclareMathOperator{\conv}{\operatorname{conv}}

\begin{document}

\title{Tight conic  approximation of testing  regions for
  quantum statistical models and measurements}

\author{Michele Dall'Arno}

\email{michele.dallarno.mv@tut.jp}

\affiliation{Department of Computer Science and Engineering,
  Toyohashi University of Technology, Japan}

\affiliation{Yukawa Institute for Theoretical Physics, Kyoto
  University, Japan}

\orcid{0000-0001-7442-4832}

\author{Francesco Buscemi}

\email{buscemi@i.nagoya-u.ac.jp}

\affiliation{Graduate  School  of  Informatics,  Nagoya
  University, Chikusa-ku, 464-8601 Nagoya, Japan}

\orcid{0000-0001-9741-0628}

\maketitle

\begin{abstract}
  Quantum statistical  models (i.e., families  of normalized
  density matrices) and quantum measurements (i.e., positive
  operator-valued measures) can be  regarded as linear maps:
  the former, mapping  the space of effects to  the space of
  probability distributions;  the latter, mapping  the space
  of states  to the space of  probability distributions. The
  images of such linear maps  are called the testing regions
  of the corresponding model or measurement. Testing regions
  are notoriously  impractical to treat analytically  in the
  quantum case. Our  first result is to  provide an implicit
  outer  approximation of  the testing  region of  any given
  quantum  statistical model  or measurement  in any  finite
  dimension:  namely, a  region  in  probability space  that
  contains  the desired  image, but  is defined  implicitly,
  using a  formula that depends  only on the given  model or
  measurement.  The outer approximation that we construct is
  \emph{minimal}  among all  such outer  approximations, and
  \emph{close},   in  the   sense   that   it  becomes   the
  \emph{maximal  inner}  approximation   up  to  a  constant
  scaling  factor.   Finally,  we  apply  our  approximation
  formulas  to characterize,  in  a semi-device  independent
  way,  the ability  to  transform  one quantum  statistical
  model or measurement into another.
\end{abstract}

\section{Introduction}

In   statistics,   information  theory,   and   mathematical
economics one is  often faced with the  problem of comparing
two  setups in  terms of  their expected  performances on  a
particular task of interest.  For example, one might compare
two statistical models by comparing their informativeness in
a given parameter estimation  problem, or two noisy channels
with respect  to a given  communication figure of  merit, or
again two portfolios with  respect to their expected utility
in a  given betting scenario.  The comparison could  also be
extended,  so to  ask when  a given  setup is  \emph{always}
better than another one, i.e., independent of any particular
task  at   hand.  Such  ``global''   comparisons,  generally
described by a preorder relation, play a crucial role in the
formulation of mathematical statistics.

The simplest example  of one such preorder  in statistics is
given  by the  \emph{majorization  preorder} of  probability
distributions~\cite{lorenz1905methods,hardy1952inequalities,Arnold87,MarshallOlkin}. Generalizing
this, we find  the comparison of families  comprising two or
more  probability  distributions.  The   case  of  pairs  of
probability distributions (i.e., \emph{dichotomies}) is also
known                   as                   \emph{relative}
majorization~\cite{blackwell1953,blackwell-girshick-theory-games,torgersen_1991,renes-relative-submajor},
whereas the case of multiple elements is usually referred to
as   comparison   of   statistical   \emph{experiments}   or
\emph{models}~\cite{blackwell1953,blackwell-girshick-theory-games,torgersen_1991,liese-miescke}.

The relevance  of such  preorder relations is  epitomized by
Blackwell's
theorem~\cite{blackwell1953,blackwell-girshick-theory-games},
which   establishes  the   equivalence  between   the  above
mentioned statistical  comparisons, and  the existence  of a
suitable  stochastic  map  that transforms  one  setup  (the
``always better'' one) into  the other (the ``always worse''
one). For this reason,  Blackwell's theorem and its variants
provide   a   powerful   framework  for   general   resource
theories~\cite{Chitambar-Gour2019resource-theories},     and
indeed    recent   quantum    extensions   of    Blackwell's
theorem~\cite{buscemi-CMP-2012,jencova-comparison-2015,buscemi2017quantum}
have  found fruitful  application  in the  study of  quantum
entanglement~\cite{buscemi2012all},                  quantum
thermodynamics~\cite{buscemi-2015-fully-quantum-second-laws,Gour:2018aa-jennings-buscemi-2018},
and                    quantum                   measurement
theory~\cite{buscemi2020complete,buscemi-2022-unifying-instrument-incompatibility,buscemi2023sharpness},
for example.

Mathematically, equivalence theorems \emph{\`a la} Blackwell
start   from  the   characterization  of   suitably  defined
\emph{testing  regions},  corresponding to  the  statistical
models at hand. In the simplest scenario, the testing region
of  a   statistical  model   $\{\rho_i:1\le  i\le   n\}$  is
constructed as follows: for any effect $0\le\pi\le\openone$,
one computes  the $n$-dimensional  real vector  whose $i$-th
component is $\Tr[\pi\;\rho_i]$; the  collection of all such
vectors, for  varying effect  $\pi$, constitute  the testing
region of $\{\rho_i\}_i$\footnote{The  definition of testing
region can be straightforwardly extended also to families of
effects  $\boldsymbol{\pi}=\{\pi_i:1\le i\le  n\}$. In  this
case  the  region  in  $\mathbb{R}^n$  to  consider  is  the
collection  of   vectors  whose  components  are   given  by
$\Tr[\pi_i\;\rho]$,  for varying  $\rho$ in  the set  of all
states.}.   In   other  words,  the  testing   region  of  a
statistical model is  the \emph{image} of the  set of effect
through  the linear  map  induced by  the  former. For  this
reason,  in what  follows we  will use  the terms  ``testing
region''  and  ``image'' interchangeably.   Two  statistical
models with the same number of elements can then be compared
by looking at their testing regions. A particularly relevant
condition occurs when the  testing region of one statistical
model  contains that  of  the  other one.   In  the case  of
dichotomies,  the  inclusion  relation for  testing  regions
corresponds   exactly   with   the  preorder   of   relative
majorization~\cite{renes-relative-submajor,buscemi2017quantum}.

Unfortunately,   due  to   the   non-commutativity  of   the
underlying  algebra,  the  quantum  version  of  Blackwell's
equivalence~\cite{buscemi-CMP-2012}  turns  out to  be  more
convoluted than  its original classical variant.  One reason
for this is that  testing regions quickly become impractical
to  treat analytically\footnote{Another  reason is  that the
requirement   of  \emph{complete   positivity}  demands   an
extended   comparison~\cite{buscemi-CMP-2012}.}.   This   is
particularly  evident  already  in   the  case  of  relative
majorization: while  classical relative majorization  can be
summarized  in  a  finite collection  of  easily  computable
inequalities~\cite{blackwell1953,renes-relative-submajor},
in  the  quantum  case   (with  the  notable  exceptions  of
qubits~\cite{alberti1982states,DallArno2020extensionofalberti})
an   infinite  number   of  scalar   inequalities  must   be
evaluated~\cite{buscemi2017quantum}.  The  situation becomes
even  more cumbersome  in  the case  of quantum  statistical
models~\cite{buscemi-CMP-2012}.

In this paper, in order to  shed more light on the structure
of  quantum  testing  regions,   we  provide  techniques  to
construct implicit  approximations of the testing  region of
arbitrary  quantum statistical  models and  measurements, in
any  finite dimension.  More precisely,  we construct  conic
regions   in   probability   space   that   contain   (outer
approximations), or are contained (inner approximations) by,
the desired testing region.  Such approximations, unlike the
testing region,  can be defined implicitly,  using a formula
that  depends  only  on   the  given  setup  (i.e.,  quantum
statistical model  or measurement). The  approximations that
we   construct    are   \emph{optimal}   among    all   such
approximations, that is, we prove  that they are the minimal
outer and  the maximal inner conic  approximations. They are
moreover \emph{close},  in the sense that  the minimal outer
approximation becomes the maximal  inner approximation up to
a constant scaling factor. Our approximation techniques thus
generalize    the    bounding     recently    provided    in
Ref.~\cite{xu2023bounding}  by  Xu, Schwonnek,  and  Winter:
first,  the extension  is  from Pauli  strings to  arbitrary
measurements; second, the optimization  is not restricted to
the   radius  of   fixed-axis  ellipsoids,   but  it   is  a
\emph{global} optimization  over all  the parameters  of the
ellispoid.  As an application,  we utilize our approximation
formulas to characterize, in  a semi-device independent way,
the  ability  to  transform  one  quantum  measurement  into
another, or one quantum statistical model into another.

\section{Main Results}

\subsection{Quantum measurements}

Given     a      $d$-dimensional     quantum     measurement
$\boldsymbol{\pi}=\{\pi_i:1\le  i\le   n\}$,  $\pi_i\ge  0$,
$\sum_i\pi_i=\openone$, its testing region is defined as the
image  $\boldsymbol{\pi}  (  \mathbb{S}_d   )$  of  the  set
$\mathbb{S}_d$    of    $d$-dimensional    states    through
$\boldsymbol{\pi}$.   By  definition,   this  is   given  in
parametric form,  that is, it  is a body in  the probability
space parameterized by states  in the state space.  Ideally,
one would  aim at implicitizing  it, that is, writing  it in
the form  $f(p) \le  1$, for probability  distributions $p$.
However,  due to  intractability  of the  strucuture of  the
state  space,   we  resort   here  to   providing  inclusion
conditions in terms of implicit bodies.

\begin{dfn}
  \label{def:image}
  For    any   $d$-dimensional,    $n$-outcome   measurement
  $\boldsymbol{\pi}=\{\pi_i\}_{i=1}^n$, we define the family
  $\{   \mathcal{E}_r  (   \boldsymbol{\pi}   )  \}_{r   \in
    \mathbb{R}}$ of hyper-ellipsoids given by:
  \begin{align*}
    \mathcal{E}_r \left( \boldsymbol{\pi} \right) := \left\{
    \mathbf{p}  \in   \boldsymbol{\pi}  \left(  \mathbb{C}^d
    \right)  \Big|  \left|  \sqrt{Q^+} \left(  \mathbf{p}  -
    \mathbf{t} \right) \right|_2^2 \le \frac1{r^2} \right\},
  \end{align*}
  where  $Q \in  \mathbb{R}^{n \times  n}$ is  the symmetric
  positive semi-definite covariance matrix given by
  \begin{align*}
    Q_{ij}  =  \frac{d-1}d   \left(  \Tr\left[  \pi_i  \pi_j
      \right] -  \frac{ \Tr\left[  \pi_i \right]  \Tr[ \pi_j
    ]}d \right),
  \end{align*}
  for  any  $0  \le  i,   j  \le  n$,  and  $\mathbf{t}  \in
  \mathbb{R}^n$ is the vector
  \begin{align*}
    t_i = \frac1d  \Tr \left[ \pi_i \right], \qquad  1 \le i
    \le n.
  \end{align*}
\end{dfn}

\begin{thm}
  \label{thm:image}
  For   any  $d$-dimensional,   $n$-outcome  informationally
  complete  measurement  $\boldsymbol{\pi}$,  one  has  that
  $\mathcal{E}_{d -  1} ( \boldsymbol{\pi})$ is  the maximum
  volume   ellipsoid   enclosed   in   $\boldsymbol{\pi}   (
  \mathbb{S}_d )$  and $\mathcal{E}_1 (  \boldsymbol{\pi} )$
  is    the     minimum    volume     ellipsoid    enclosing
  $\boldsymbol{\pi} ( \mathbb{S}_d )$.
\end{thm}

If  measurement  $\boldsymbol{\pi}$ is  not  informationally
complete,    ellipsoids     $\mathcal{E}_{d    -     1}    (
\boldsymbol{\pi})$  and $\mathcal{E}_1  ( \boldsymbol{\pi})$
still    are    inner    and   outer    approximations    of
$\boldsymbol{\pi}   (   \mathbb{S}_d    )$,   although   not
necessarily maximal and minimal in volume, respectively.

We   postpone  the   proof  of   Theorem~\ref{thm:image}  to
Section~\ref{sect:proof}.

As  examples,  let  us consider  symmetric,  informationally
complete   (SIC)   and   mutually   unbiased   basis   (MUB)
measurements.

A $d$-dimensional  measurement $\boldsymbol{\pi}$ is  SIC if
and  only  if  it  has  $n =  d^2$  effects  satisfying  the
condition $\Tr  \pi_i \pi_j = (  d \delta_{i,j} + 1)  / (d^2
(d+1))$. By explicit computation one has
\begin{align*}
  Q  =   \frac{d-1}{d^2  \left(   d  +  1   \right)}  \left(
  \openone_{d^2}   -   \hat{\mathbf{u}}   \hat{\mathbf{u}}^T
  \right).
\end{align*}
As expected, $Q$ is a $d^2 \times d^2$ matrix of rank $d^2 -
1$,      and      it      is     proportional      to      a
projector~\cite{Slomczynski2020morphophoricpovms}.       Its
pseudo-inverse is then given by
\begin{align*}
  Q^+  =  \frac{d^2  \left(   d  +  1  \right)}{d-1}  \left(
  \openone_{d^2}   -   \hat{\mathbf{u}}   \hat{\mathbf{u}}^T
  \right).
\end{align*}

A  $d$-dimensional   measurement  $\boldsymbol{\pi}$   is  a
complete MUB if and  only if it has $n = d  (d + 1)$ effects
satisfying  the  condition  $\Tr  \pi_{i, j}  \pi_{k,  l}  =
(\delta_{i, k} \delta{j, l} + (1  - \delta_{i, k})/d) / (d +
1)^2$, where indices $i, k$ denote the basis and indices $j,
l$  denote  the  effect   within  the  basis.   By  explicit
copmputation one has
\begin{align*}
  Q  =  \frac{d  -  1}{d  \left( d  +  1  \right)^2}  \left(
  \openone_{d \left( d  + 1 \right)} -  \oplus_{i = 1}^{d+1}
  \hat{\mathbf{u}}_d^i \hat{\mathbf{u}}_d^{i T} \right),
\end{align*}
where  $\mathbf{u}_d^i$  is the  vector  with  ones for  the
entries corresponding  to basis  $i$ and zero  otherwise. As
expected, $Q$ is  a $d (d+1) \times d (d+1)$  matrix of rank
$d^2    -    1$,   and    it    is    proportional   to    a
projector~\cite{Slomczynski2020morphophoricpovms}.       Its
pseudo-inverse is then given by
\begin{align*}
  Q^+  = \frac{d  \left(  d  + 1  \right)^2}{d  - 1}  \left(
  \openone_{d \left( d  + 1 \right)} -  \oplus_{i = 1}^{d+1}
  \hat{\mathbf{u}}_d^i \hat{\mathbf{u}}_d^{i T} \right).
\end{align*}

Now that we  have a close approximation of the  image of the
set of  states through  any given  measurement, we  turn our
attention to applying it to semi-device independent tests of
simulability.  A test is  semi-device independent if it only
assumes the dimension of the  devices involved, but does not
otherwise assume their mathematical description. We say that
a      $d_1$-dimensional,       $n$-outcome      measurement
$\boldsymbol{\pi}_1$    simulates    a    $d_0$-dimensional,
$n$-outcome measurement $\boldsymbol{\pi}_0$  if and only if
there  exists  a  completely  positive  map  $\mathcal{C}  :
\mathcal{L}    (    \mathbb{C}^{d_0})    \to\mathcal{L}    (
\mathbb{C}^{d_1} )$ such that
\begin{align}
  \label{eq:simulation}
  \boldsymbol{\pi}_1 \circ \mathcal{C} = \boldsymbol{\pi}_0.
\end{align}

The   following  corollary   generalizes   Corollary  2   of
Ref.~\cite{DallArno2020extensionofalberti} to  the arbitrary
dimensional case,  providing a semi-device  independent test
of   Eq.~\eqref{eq:simulation}.

\begin{cor}[Semi-device independent simulability test]
  Given  a  set  $\mathcal{P}$  of  $n$-element  probability
  distributions generated by  a $d_1$-dimensional (otherwise
  unspecified)  measurement  $\boldsymbol{\pi}_1$,  for  any
  $d_0$   and   for    any   $d_0$-dimensional   $n$-outcome
  measurement $\boldsymbol{\pi}_0$ such that
  \begin{align*}
    \mathcal{E}_1    \left(    \boldsymbol{\pi}_0    \right)
    \subseteq \conv \mathcal{P},
  \end{align*}
  there exists a trace  preserving map $\mathcal{C}$ that is
  positive on the support  of $\boldsymbol{\pi}_0$ such that
  Eq.~\eqref{eq:simulation} holds.  Moreover, if $D = 2$, $n
  \le   3$,   and   $d   \le  3$,   map   $\mathcal{C}$   in
  Eq.~\eqref{eq:simulation} is completely positive, that is,
  measurement  $\boldsymbol{\pi}_1$   simulates  measurement
  $\boldsymbol{\pi}_0$.
\end{cor}

\begin{proof}
  The   first   part   of   the   statement   follows   from
  Theorem~\ref{thm:image}   and   from  Proposition~7.1   of
  Ref.~\cite{buscemi-2005-clean-POVMs}.  The  second part of
  the  statement  follows from  Theorem~\ref{thm:image}  and
  from               Theorem              2               of
  Ref.~\cite{DallArno2020extensionofalberti}.
\end{proof}

\subsection{Quantum statistical models}

Given   a   $d$-dimensional    quantum   statistical   model
$\boldsymbol{\rho}=\{\rho_i:1\le  i\le n\}$,  $\rho_i\ge 0$,
$\Tr[\rho_i]=1$, its testing region  is defined as the image
$\boldsymbol{\rho} ( \mathbb{E} )$  of the cone $\mathbb{E}$
of   effects   through   $\boldsymbol{\rho}$,  seen   as   a
classical-quantum  (c-q for  short) channel.  By definition,
the  testing region  $\boldsymbol{\rho} (  \mathbb{E} )$  is
given  in parametric  form, that  is, it  is a  body in  the
probability  space parameterized  by effects  in the  effect
space. Ideally, one would aim  at implicitizing it, that is,
write  it  in  the  form   $f(q)  \le  1$,  for  vectors  of
probabilities $q$. However, due to the intractability of the
structure of the  effect space, we resort  here to providing
inclusion conditions in terms of implicit bodies.

\begin{dfn}
  \label{def:image2}
  For          any           $d$-dimensional          family
  $\boldsymbol{\rho}=\{\rho_i\}_{i=1}^n$ of  $n$ states, let
  $\{\mathcal{E}_r^k  (  \boldsymbol{\rho})  \}^{k=0,  \dots
    d}_{r  \in  \mathbb{R}}$  be  the  following  family  of
  hyper-ellipsoids:
  \begin{align*}
    \mathcal{E}_r^k   \left(  \boldsymbol{\rho}   \right)  =
    \left\{   \mathbf{q}    \in   \boldsymbol{\rho}   \left(
    \mathbb{C}^d  \right) \Big|  \left| \sqrt{Q_k^+}  \left(
    \mathbf{q}  - \frac{k}d  \mathbf{u} \right)  \right|_2^2
    \le \frac1{r^2} \right\},
  \end{align*}
  where $Q_k  \in \mathbb{R}^{n \times n}$  is the symmetric
  positive semi-definite covariance matrix given by
  \begin{align*}
    \left( Q_k \right)_{ij} = \left( k - \frac{k^2}d \right)
    \left(  \Tr\left[   \rho_i  \rho_j  \right]   -  \frac1d
    \right),
  \end{align*}
  for  any  $0  \le  i,   j  \le  n$,  and  $\mathbf{u}  \in
  \mathbb{R}^n$ is the vector with all unit entries.
\end{dfn}

We introduce a $d$-cone as a generalization of the bicone. A
$d$-cone in $\mathbb{R}^n$ is the  convex hull of the origin
and $d$  arbitrary $(n-1)$-balls with aligned  centers lying
on hyperplanes  orthogonal to the  line of the  centers. Let
$r(x)$ be  the radius of the  ball at distance $x$  from the
origin  and $L$  be the  distance of  the furthest  ball. If
$r(x)$ is  symmetric, that is $r(x)  = r(L - x)$  for any $0
\le x  \le L$, then we  say that the $d$-cone  is symmetric.
The usual bicone is recovered  as the symmetric $2$-cone.  A
pictorial   representation   of   $d$-cones  is   given   in
Fig.~\ref{fig:dcones}. An  elliptical $d$-cone is  the image
of a $d$-cone through a linear transformation that preserves
the line joining the centers of the balls.
\begin{figure}[h!]
  \begin{center}
    \includegraphics[width=\columnwidth]{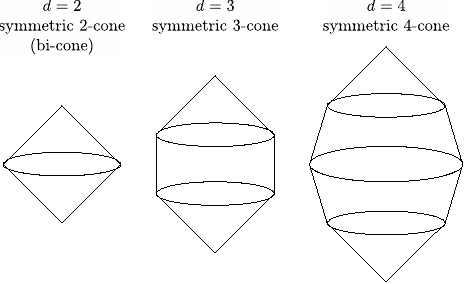}
  \end{center}
  \caption{A pictorial representation of symmetric $d$-cones
    in $\mathbb{R}^3$, for $d = 2, 3, 4$.}
  \label{fig:dcones}
\end{figure}

\begin{thm}
  \label{thm:image2}
  For   any  $d$-dimensional,   $n$-outcome  informationally
  complete  family $\boldsymbol{\rho}$  of  states, one  has
  that   $\conv  \cup_{k=0}^d   \mathcal{E}_{d   -  1}^k   (
  \boldsymbol{\rho})$  is  the   maximum  volume  elliptical
  $d$-cone  enclosed in  $\boldsymbol{\rho} (  \mathbb{E} )$
  and     $\conv      \cup_{k=0}^d     \mathcal{E}_1^k     (
  \boldsymbol{\rho})$  is  the   minimum  volume  elliptical
  $d$-cone enclosing $\boldsymbol{\rho} ( \mathbb{E} )$.
\end{thm}

If   family    $\boldsymbol{\rho}$   of   states    is   not
informationally   complete,    elliptical   d-cones   $\conv
\cup_{k=0}^d \mathcal{E}_{d - 1}^k ( \boldsymbol{\rho})$ and
$\conv  \cup_{k=0}^d  \mathcal{E}_1^k (  \boldsymbol{\rho})$
still    are    inner    and   outer    approximations    of
$\boldsymbol{\rho} ( \mathbb{E} )$, although not necessarily
maximal and minimal in volume, respectively.

We  postpone   the  proof  of   Theorem~\ref{thm:image2}  to
Section~\ref{sect:proof2}.

As  examples,  let  us consider  symmetric,  informationally
complete   (SIC)   and   mutually   unbiased   basis   (MUB)
families of states.

A  $d$-dimensional family  $\boldsymbol{\rho}$ of  states is
SIC if  and only if it  has $n = d^2$  states satisfying the
condition  $\Tr \rho_i  \rho_j =  (  d \delta_{i,j}  + 1)  /
(d+1)$. By explicit computation one has
\begin{align*}
  Q_k =  \frac{k d-  k^2}{d^2 \left( d  + 1  \right)} \left(
  \openone_{d^2}   -   \hat{\mathbf{u}}   \hat{\mathbf{u}}^T
  \right).
\end{align*}
As expected, $Q_k$ are $d^2  \times d^2$ matrix of rank $d^2
- 1$,  and  they are  proportional  to  a projector.   Their
pseudo-inverses are then given by
\begin{align*}
  Q^+_k = \frac{d^2 \left( d +  1 \right)}{k d - k^2} \left(
  \openone_{d^2}   -   \hat{\mathbf{u}}   \hat{\mathbf{u}}^T
  \right).
\end{align*}

A $d$-dimensional family $\boldsymbol{\rho}$  of states is a
complete MUB if  and only if it  has $n = d (d  + 1)$ states
satisfying  the condition  $\Tr  \rho_{i, j}  \rho_{k, l}  =
(\delta_{i, k} \delta{j, l} + (1 - \delta_{i, k})/d)$, where
indices $i,  k$ denote the  basis and indices $j,  l$ denote
the effect  within the basis.  By  explicit copmputation one
has
\begin{align*}
  Q_k =  \frac{k d - k^2}{d  \left( d + 1  \right)^2} \left(
  \openone_{d \left( d  + 1 \right)} -  \oplus_{i = 1}^{d+1}
  \hat{\mathbf{u}}_d^i \hat{\mathbf{u}}_d^{i T} \right),
\end{align*}
where  $\mathbf{u}_d^i$  is the  vector  with  ones for  the
entries corresponding  to basis  $i$ and zero  otherwise. As
expected, $Q_k$ are $d (d+1) \times d (d+1)$ matrces of rank
$d^2 - 1$, and they  are proportional to a projector.  Their
pseudo-inverses are then given by
\begin{align*}
  Q^+_k = \frac{d \left( d +  1 \right)^2}{k d - k^2} \left(
  \openone_{d \left( d  + 1 \right)} -  \oplus_{i = 1}^{d+1}
  \hat{\mathbf{u}}_d^i \hat{\mathbf{u}}_d^{i T} \right).
\end{align*}

Now that we  have a close approximation of the  image of the
set of effects  through any given family of  states, we turn
our  attention to  applying  it  to semi-device  independent
tests  of simulability.   We say  that a  $d_1$-dimensional,
$n$-outcome family of states $\boldsymbol{\rho}$ simulates a
$d_0$-dimensional,          $n$-outcome          measurement
$\boldsymbol{\rho}_0$  if   and  only  if  there   exists  a
completely positive trace preserving map (a quantum channel)
$\mathcal{C}    :   \mathcal{L}    (   \mathbb{C}^{d_1}    )
\to\mathcal{L} ( \mathbb{C}^{d_0} )$ such that
\begin{align}
  \label{eq:simulation2}
  \mathcal{C}       \circ        \boldsymbol{\rho}_1       =
  \boldsymbol{\rho}_0.
\end{align}
The   following  corollary   generalizes   Corollary  1   of
Ref.~\cite{DallArno2020extensionofalberti} to  the arbitrary
dimensional case,  providing a semi-device  independent test
of Eq.~\eqref{eq:simulation2}.

\begin{cor}[Semi-device independent simulability test]
  Given  a  set  $\mathcal{Q}$  of  $n$-element  vectors  of
  probabilities generated by  a $d_1$-dimensional (otherwise
  unspecified) family  of $n$  states $\boldsymbol{\rho}_1$,
  for any $d_0$ and for  any $d_0$-dimensional family of $n$
  states $\boldsymbol{\rho}_0$ such that
  \begin{align*}
    \conv      \cup_{k=0}^d      \mathcal{E}_1^k      \left(
    \boldsymbol{\rho}_0 \right) \subseteq \conv \mathcal{Q},
  \end{align*}
  there  exists a  (not  necessarily  trace preserving)  map
  $\mathcal{C}$  that   is  positive   on  the   support  of
  $\boldsymbol{\rho}_0$ such  that Eq.~\eqref{eq:simulation}
  holds.  Moreover, if  $D = 2$, $n  = 2$, and $d  = 2$, map
  $\mathcal{C}$  in Eq.~\eqref{eq:simulation}  is completely
  positive    trace    preserving,     that    is,    family
  $\boldsymbol{\rho}_1$    of   states    simulates   family
  $\boldsymbol{\rho}_0$ of states.
\end{cor}

\begin{proof}
  The   first   part   of   the   statement   follows   from
  Theorem~\ref{thm:image2}.    The   second  part   of   the
  statement follows  from Theorem~\ref{thm:image2}  and from
  Theorem 1 of Ref.~\cite{DallArno2020extensionofalberti}.
\end{proof}

\section{Proofs}

\setcounter{thm}{0}

\subsection{Formalization}
\label{sect:formalization}

For   any  positive   integer   $d$,   let  $\mathcal{L}   (
\mathbb{C}^d )$  denote the space of  Hermitian operators on
$\mathbb{C}^d$  equipped with  the Hilbert-Schmidt  product,
that is, for  any $\rho, \pi \in  \mathcal{L} ( \mathbb{C}^d
)$ we  have $\rho  \cdot \pi =  \Tr [ \rho  \pi ]$.  For any
positive integer $n$, let $\mathbb{R}^n$ denote the space of
$n$-dimensional real  vectors equipped with the  usual inner
product, that is, for any $p, q \in \mathbb{R}^n$ we have $p
\cdot q = \sum_{i = 1}^n p_i^\dagger q_i$.

A   $d$-dimensional,  $n$-outcome   measurement  is   a  map
\begin{align*}
  \boldsymbol{\pi} : \mathcal{L} \left( \mathbb{C}^d \right)
  \to \mathbb{R}^n.
\end{align*}

Any measurement $\boldsymbol{\pi}$ can  be represented as an
indexed family  $\{ \pi_i  \in \mathcal{L} (  \mathbb{C}^d )
\}_{i=1}^n$  of operators  as  follows.  Recalling that  the
space $\mathcal{L}  ( \mathbb{C}^d  )$ is equipped  with the
Hilbert-Schmidt product, the action of $\boldsymbol{\pi}$ on
an  operator  $\rho \in  \mathcal{L}  (  \mathbb{C}^d )$  is
naturally given by
\begin{align*}
  \boldsymbol{\pi} \left( \rho \right) :=
  \begin{bmatrix}
    \bbra{\pi_1} \\
    \vdots \\
    \bbra{\pi_n}
  \end{bmatrix} \kett{\rho} =   \begin{bmatrix}
    \Tr \left[ \pi_1 \rho \right] \\
    \vdots \\
    \Tr \left[ \pi_n \rho \right]
  \end{bmatrix} \in \mathbb{R}^n,
\end{align*}
where  $\bbra{\pi}  :  \mathcal{L}   (  \mathbb{C}^d  )  \to
\mathbb{R}$  is  given  by  $\bbrakett{\pi|\rho}  =  \Tr[\pi
  \rho]$.

Recalling that the space  $\mathbb{R}^n$ is instead equipped
with the  usual inner product,  the action of  the Hermitian
conjugate $\boldsymbol{\pi}^\dagger$ on a vector $\mathbf{p}
\in \mathbb{R}^n$ is naturally given by
\begin{align*}
  \boldsymbol{\pi}^\dagger \mathbf{p} = &
  \begin{bmatrix}
    \kett{\pi_1} & \dots & \kett{\pi_n}
  \end{bmatrix}
  \begin{bmatrix}
    p_1 \\ \vdots \\ p_n
  \end{bmatrix}\\
  = & \sum_{i = 1}^n p_i \kett{\pi_i} \in \mathcal{L} \left(
  \mathbb{C}^d \right).
\end{align*}
Finally,   for  any   measurements  $\boldsymbol{\pi}$   and
$\boldsymbol{\tau}$,             the            compositions
$\boldsymbol{\tau}^\dagger    \boldsymbol{\pi}$     and    $
\boldsymbol{\pi} \boldsymbol{\tau}^\dagger$ are given by
\begin{align*}
  \boldsymbol{\tau}^\dagger \boldsymbol{\pi} & =
  \begin{bmatrix}
    \kett{\tau_1} & \dots & \kett{\tau_n}
  \end{bmatrix}
  \begin{bmatrix}
    \bbra{\pi_1} \\
    \vdots \\
    \bbra{\pi_n}
  \end{bmatrix}\\  & =
  \sum_{i   =  1}^n   \kett{\tau_i}  \!    \bbra{\pi_i}  \in
  \mathcal{L}  \left( \mathbb{C}^d  \right) \to  \mathcal{L}
  \left( \mathbb{C}^d \right),
\end{align*}
and
\begin{align*}
  \boldsymbol{\pi} \boldsymbol{\tau}^\dagger & =
  \begin{bmatrix}
    \bbra{\pi_1} \\
    \vdots \\
    \bbra{\pi_n}
  \end{bmatrix}
  \begin{bmatrix}
    \kett{\tau_1} & \dots & \kett{\tau_n}
  \end{bmatrix}\\
  & =
  \begin{bmatrix}
    \Tr \left[ \pi_1 \tau_1 \right] & \dots & \Tr \left[ \pi_1 \tau_n \right]\\
    \vdots & & \vdots\\
    \Tr \left[ \pi_n \tau_1 \right] & \dots & \Tr \left[ \pi_n \tau_n \right]
  \end{bmatrix}
  \in  \mathbb{R}^n \to \mathbb{R}^n.
\end{align*}

For    any    $d$-dimensional,    $n$-outcome    measurement
$\boldsymbol{\pi}$, its  pseudo-inverse $\boldsymbol{\pi}^+$
is  the  unique  $n$-elements  row vector  of  operators  in
$\mathcal{L} ( \mathbb{C}^d )$ such that
\begin{align*}
  \boldsymbol{\pi}  \boldsymbol{\pi}^+ \boldsymbol{\pi}  & =
  \boldsymbol{\pi},\\   \boldsymbol{\pi}^+   \boldsymbol{\pi}
  \boldsymbol{\pi}^+                   &                   =
  \boldsymbol{\pi}^+,\\ \boldsymbol{\pi}^+ \boldsymbol{\pi} &
  =      \left(     \boldsymbol{\pi}^+      \boldsymbol{\pi}
  \right)^\dagger,\\ \boldsymbol{\pi}  \boldsymbol{\pi}^+ & =
  \left( \boldsymbol{\pi} \boldsymbol{\pi}^+ \right)^\dagger.
\end{align*}

\subsection{Quantum measurements}
\label{sect:proof}

Leveraging     on     the    formalism     introduced     in
Section~\ref{sect:formalization},  for any  $d$-dimensional,
$n$-outcome  measurement $\boldsymbol{\pi}$  we can  provide
the  following  definitions  of covariance  matrix  $Q$  and
probability distribution $\mathbf{t}$:
\begin{align*}
  Q    :=    \frac{d-1}d     \left(    \boldsymbol{\pi}    -
  \boldsymbol{\tau}   \right)   \left(  \boldsymbol{\pi}   -
  \boldsymbol{\tau} \right)^\dagger,
\end{align*}
and
\begin{align*}
  \mathbf{t} := \boldsymbol{\tau} \frac{\kett{\openone}}d,
\end{align*}
where $\boldsymbol{\tau}$  is  the $d$-dimensional,  $n$-outcome
measurement given by
\begin{align*}
  \boldsymbol{\tau}  := \frac1d  \begin{bmatrix} \Tr  \left[
      \pi_1  \right] \bbra{\openone}\\  \vdots\\ \Tr  \left[
      \pi_n \right] \bbra{\openone}
  \end{bmatrix}.
\end{align*}
Notice that  these definitions  are consistent with  those in
Def.~\ref{def:image}.

For any dimension $d$ we denote with $\mathbb{B}_d$ the ball
whose extremal points include all pure states, that is
\begin{align*}
  \mathbb{B}_d  :=  \left\{   \rho  \in  \mathcal{L}  \left(
  \mathbb{C}^d \right) \Big| \Tr \left[ \rho \right] = 1, \;
  \Tr \left[ \rho^2 \right] \le 1 \right\}.
\end{align*}
Consider the  image $\boldsymbol{\pi}  ( \mathbb{B}_d  )$ of
the    ball    $\mathbb{B}_d$    through    a    measurement
$\boldsymbol{\pi}$.  Again, this expression describes a body
in  the probability  space parameterized  by a  body in  the
state  space.   The  following  lemma  makes  implicit  this
parametric equation by removing the dependence on the states
and  expressing  the image  of  $\mathbb{B}_d$  in the  form
$f(\mathbf{p}) \le  0$. The  lemma generalizes  Theorem~1 of
Ref.~\cite{DBBV17}  from the  qubit  case  to the  arbitrary
dimensional case.

\begin{lmm}[Implicitization of $\boldsymbol{\pi} (
    \mathbb{B}_d  )$]
  \label{lmm:implicitization}
  For    any   $d$-dimensional,    $n$-outcome   measurement
  $\boldsymbol{\pi}$,   the    image   $\boldsymbol{\pi}   (
  \mathbb{B}_d )$ is given by the following hyper-ellipsoid:
  \begin{align*}
    \boldsymbol{\pi}   \left(    \mathbb{B}_d   \right)   :=
    \mathcal{E}_1 \left( \boldsymbol{\pi} \right).
  \end{align*}
\end{lmm}

\begin{proof}
  One  has
  \begin{align*}
    \mathbf{p}  &  =  \boldsymbol{\pi} \kett{\rho}  \\  &  =
    (\boldsymbol{\pi}       -      \boldsymbol{\tau}       +
    \boldsymbol{\tau}) \kett{\rho} \\  & = (\boldsymbol{\pi}
    - \boldsymbol{\tau})  \kett{\rho}   +  \boldsymbol{\tau}
    \kett{\rho}     \\    &     =    (\boldsymbol{\pi}     -
    \boldsymbol{\tau}) \kett{\rho} + \mathbf{t}.
  \end{align*}
  Hence
  \begin{align*}
    \boldsymbol{\pi} \left(  \mathbb{B}_d \right)  = \left\{
    \mathbf{p} = \left( \boldsymbol{\pi} - \boldsymbol{\tau}
    \right) \kett{\rho} + \mathbf{t} \Big| \Tr \rho = 1, \Tr
    \rho^2 \le 1 \right\}.
  \end{align*}

  Solutions  of $(  \boldsymbol{\pi}  - \boldsymbol{\tau}  )
  \kett{ \rho }  = \mathbf{p} - \mathbf{t}$  in $\rho$ exist
  if and  only if $\mathbf{p}  - \mathbf{t}$ belongs  to the
  range of $\boldsymbol{\pi} - \boldsymbol{\tau}$. Solutions
  are given by
  \begin{align}
    \label{eq:solutions}
    \kett{\rho} = ( \boldsymbol{\pi} - \boldsymbol{\tau} )^+
    \left( \mathbf{p} - \mathbf{t} \right) + \left( \openone
    - \Pi \right) \kett{\sigma},
  \end{align}
  where $\Pi := ( \boldsymbol{\pi} - \boldsymbol{\tau} )^+ (
  \boldsymbol{\pi} -  \boldsymbol{\tau} )$, for  any $\sigma
  \in  \mathcal{L}  (  \mathbb{C}^d )$.   Notice  that  $\Pi
  \kett{\openone}   =  0$   since   $(  \boldsymbol{\pi}   -
  \boldsymbol{\tau})   \kett{\openone}    =   \mathbf{t}   -
  \mathbf{t}$. Hence  Eq.~\eqref{eq:solutions} is equivalent
  to
  \begin{align*}
    \kett{\rho} =  & ( \boldsymbol{\pi}  - \boldsymbol{\tau}
    )^+  \left( \mathbf{p}  - \mathbf{t}  \right) +  \lambda
    \frac{\kett{\openone}}d\\ & + \left(  \openone -  \frac1d
    \kett{\openone}  \!\!   \bbra{\openone}  -  \Pi  \right)
    \kett{\sigma},
  \end{align*}
  again for any $\sigma \in \mathcal{L} ( \mathbb{C}^d )$.

  The condition $\Tr \rho  = 1$ immediately implies $\lambda
  = 1$.  Moreover, due  to the Hilbert-Schmidt orthogonality
  of   $(  \boldsymbol{\pi}   -   \boldsymbol{\tau}  )^+   (
  \mathbf{p}   -   \mathbf{t}   )$   and   $(   \openone   -
  \kett{\openone}   \!\!    \bbra{\openone}/d   -    \Pi   )
  \kett{\sigma}$, one  has that  for any $\sigma$  such that
  $\Tr \rho^2  \le 1$, the  same condition is  also verified
  for $\sigma  = 0$.  Hence, without  loss of  generality we
  take $\sigma = 0$. Thus we have
  \begin{align*}
    \kett{\rho}     =      \left(     \boldsymbol{\pi}     -
    \boldsymbol{\tau}   \right)^+    \left(   \mathbf{p}   -
    \mathbf{t} \right) + \frac{\kett{\openone}}d.
  \end{align*}
  
  Hence,
  \begin{align*}
    \Tr \rho^2  = \left(  \mathbf{p} -  \mathbf{t} \right)^T
    \left(     \boldsymbol{\pi}      -     \boldsymbol{\tau}
    \right)^{+\dagger}     \left(     \boldsymbol{\pi}     -
    \boldsymbol{\tau}   \right)^+    \left(   \mathbf{p}   -
    \mathbf{t} \right) + \frac1d.
  \end{align*}
  Thus, condition $\Tr \rho^2 \le 1$ becomes
  \begin{align*}
    \left(   \mathbf{p}   -  \mathbf{t}   \right)^T   \left(
    \boldsymbol{\pi} -  \boldsymbol{\tau} \right)^{+\dagger}
    \left(  \boldsymbol{\pi}  - \boldsymbol{\tau}  \right)^+
    \left( \mathbf{p} - \mathbf{t} \right) \le 1 - \frac1d.
  \end{align*}
  Hence the statement follows.
\end{proof}

We are  now in a position  to prove Theorem~\ref{thm:image},
that we rewrite here for convenience.

\begin{thm}
  For   any  $d$-dimensional,   $n$-outcome  informationally
  complete  measurement  $\boldsymbol{\pi}$,  one  has  that
  $\mathcal{E}_{d -  1} ( \boldsymbol{\pi})$ is  the maximum
  volume   ellipsoid   enclosed   in   $\boldsymbol{\pi}   (
  \mathbb{S}_d )$  and $\mathcal{E}_1 (  \boldsymbol{\pi} )$
  is    the     minimum    volume     ellipsoid    enclosing
  $\boldsymbol{\pi} ( \mathbb{S}_d )$.
\end{thm}

\begin{proof}
  First,  we  prove  that   the  image  $\boldsymbol{\pi}  (
  \mathbb{B}_d   )$  coincides   with  the   minimum  volume
  ellipsoid $\mathcal{E} ( \boldsymbol{\pi} ( \mathbb{S}_d )
  ) $  enclosing the image  of $\mathbb{S}_d$.  This  can be
  shown  as  follows. First,  we  show  that any  $2$-design
  $\{\lambda_k, \rho_k \}_k$ is a scalable frame, that is, a
  family of weights over states such that
  \begin{align*}
    \sum_k \lambda_k \kett{\rho_k - \frac{\openone}d} \!  \!
    \bbra{\rho_k  - \frac{\openone}d}  =  \left( \openone  -
    \frac1d \kett{\openone} \!  \!  \bbra{\openone} \right).
  \end{align*}
  Indeed, for any state $\rho$ we have
  \begin{align*}
    &  \sum_k  \lambda_k  \left( \rho_k  -  \frac{\openone}d
    \right)  \Tr  \left[  \left( \rho_k  -  \frac{\openone}d
      \right)     \left(     \rho     -     \frac{\openone}d
      \right)\right]\\  =  &  \sum_k  \lambda_k  \rho_k  \Tr
    \left[   \rho_k    \left(   \rho    -   \frac{\openone}d
      \right)\right]\\  = &  \Tr_2  \left[ \sum_k  \lambda_k
      \rho_k^{\otimes 2} \left( \openone \otimes \left( \rho
      - \frac{\openone}d  \right)  \right)  \right] \\  =  &
    \Tr_2 \left[ \left( \openone + S \right) \left( \openone
      \otimes \left( \rho - \frac{\openone}d \right) \right)
      \right] \\ = & \Tr_2  \left[ S \left( \openone \otimes
      \left( \rho - \frac{\openone}d \right) \right) \right]
    \\ = & \left( \rho - \frac{\openone}d \right),
  \end{align*}
  where  $S$ denotes  the swap  operator. Notice  that, from
  Sections          6.9         and          6.11         of
  Ref.~\cite{waldron2018introduction} it immediately follows
  that finite $2$-designs exist  in any dimension $d$, hence
  the  existence of  scalable frames  in any  dimension $d$.
  Then, the statement immediately  follows from Theorem 2.11
  of Ref.~\cite{CKOPW14}.
  
  Notice  that,  if  rescaled by  constant  factor  $d^2-1$,
  minimum volume  enclosing ellipsoids  are enclosed  in the
  convex     body     (see    e.g.      Section~8.4.1     of
  Ref.~\cite{boyd2004convex}). However,  the lower  bound in
  Theorem~\ref{thm:image}  is tighter  than this,  hence the
  need for the following independent proof.

  The   inner  ellipsoid   must  include   boundary  states,
  otherwise  it would  not  maximize the  volume. Among  all
  boundary states,  the ones that minimize  the $2$-norm are
  the  projectors of  rank $d-1$.   Since $\Tr  [ \ket{\phi}
    \!\! \bra{\phi}  ]^{1/2} =  1$ and $\Tr  [ \openone  / d
  ]^{1/2} =  1 / \sqrt{d}$, one  has that the radius  of the
  outer ellipsoid  is given by  $\sqrt{1 - 1/d}  = \sqrt{(d-
    1)/  d }$.   Since  $\Tr [  \openone  - \ket{\phi}  \!\!
    \bra{\phi} /  (d-1)^2 ]^{1/2} =  1 / \sqrt{d -  1}$, one
  has that  the radius  of the inner  ellipsoid is  given by
  $\sqrt{1 /  (d-1) -  1/d} =  \sqrt{1 /  ( d  (d- 1)  ) }$.
  Hence, the  ratio of the two  radii is $\sqrt{(d- 1)/  d }
  \sqrt{ ( d (d- 1) ) } = d - 1$.
  
  Using Theorem~[J] of Ref.~\cite{Ball:1992aa}, we have that
  the lower bound in  Theorem~\ref{thm:image} holds again in
  any  dimension in  which  there exists  a finite  scalable
  frame $\{ \lambda_k, \rho_k  \}$ of states proportional to
  rank-$(d-1)$ projectors.  Since for any pure  state $\phi$
  one has
  \begin{align*}
    \frac{\openone   -    \ket{\phi}\!\!\bra{\phi}}{d-1}   -
    \frac{\openone}d     =     -    \frac{d}{d-1}     \left(
    \ket{\phi}\!\!\bra{\phi} - \frac{\openone}d \right),
  \end{align*}
  one has that such a scalable frame exists if and only if a
  scalable frame of  pure states exists, hence  the proof of
  the lower bound goes along that of the upper bound.
\end{proof}

\subsection{Quantum states}
\label{sect:proof2}

Leveraging     on     the    formalism     introduced     in
Section~\ref{sect:formalization},  for any  $d$-dimensional,
$n$-outcome  family  $\boldsymbol{\rho}$  of states  we  can
provide the following definition of covariance matrix $Q$:
\begin{align*}
  Q_k   :=   \left(   k   -   \frac{k^2}d   \right)   \left(
  \boldsymbol{\rho}  -  \boldsymbol{\sigma}  \right)  \left(
  \boldsymbol{\rho} - \boldsymbol{\sigma} \right)^\dagger,
\end{align*}
where   $\boldsymbol{\sigma}$    is   the   $d$-dimensional,
$n$-outcome c-q channel given by
\begin{align*}
  \boldsymbol{\sigma}     :=     \frac1d     \begin{bmatrix}
    \bbra{\openone}\\ \vdots\\ \bbra{\openone}
    \end{bmatrix}.
\end{align*}
Notice  that  this definition  is  consistent  with that  in
Def.~\ref{def:image2}.

For any dimension $d$ and any $0 \le k \le d$ we denote with
$\mathbb{B}_d^k$ the ball whose  extremal points include all
extremal effects with trace $k$, that is
\begin{align*}
  \mathbb{B}_d^k  :=  \left\{  \pi  \in  \mathcal{L}  \left(
  \mathbb{C}^d \right) \Big| \Tr \left[  \pi \right] = k, \;
  \Tr \left[ \pi^2 \right] \le k \right\}.
\end{align*}
We denote  with $\mathbb{D}_d$ the symmetric  $d$-cone whose
extremal points include all extremal effects, that is
\begin{align*}
  \mathbb{D}_d := \conv \cup_{k = 0}^d \mathbb{B}_d^k.
\end{align*}
Consider the  image $\boldsymbol{\rho} ( \mathbb{D}_d  )$ of
the   $d$-cone   $\mathbb{D}_d$   through  a   c-q   channel
$\boldsymbol{\rho}$.   Again,  this expression  describes  a
body in the probability space parameterized by a body in the
effect  space.   The  following lemma  makes  implicit  this
parametric  equation  by  removing  the  dependence  on  the
effects and  expressing the  image of $\mathbb{D}_d$  in the
form   $f(\mathbf{q})  \le   0$.    The  lemma   generalizes
Proposition~2 of  Ref.~\cite{Dal19} from  the qubit  case to
the arbitrary dimensional case.

\begin{lmm}[Implicitization of $\boldsymbol{\rho} (
    \mathbb{D}_d  )$]
  \label{lmm:implicitization2}
  For   any   $d$-dimensional,   $n$-outcome   c-q   channel
  $\boldsymbol{\rho}$,   the   image  $\boldsymbol{\rho}   (
  \mathbb{D}_d )$ is  given by the following  convex hull of
  hyper-ellipsoids:
  \begin{align*}
    \boldsymbol{\rho} \left(  \mathbb{D}_d \right)  := \conv
    \cup_{k=0}^d  \mathcal{E}_1^k  \left(  \boldsymbol{\rho}
    \right).
  \end{align*}
\end{lmm}

\begin{proof}
  One  has
  \begin{align*}
    \mathbf{q}  &  =  \boldsymbol{\rho} \kett{\pi}  \\  &  =
    (\boldsymbol{\rho}      -     \boldsymbol{\sigma}      +
    \boldsymbol{\sigma})     \kett{\pi}      \\     &     =
    (\boldsymbol{\rho}  -  \boldsymbol{\sigma})  \kett{\pi}  +
    \boldsymbol{\sigma} \kett{\pi} \\ & = (\boldsymbol{\rho} -
    \boldsymbol{\sigma}) \kett{\pi} + \frac{k}d \mathbf{u}.
  \end{align*}
  Hence
  \begin{align*}
    & \boldsymbol{\rho} \left( \mathbb{B}_d^k \right) \\ = &
    \left\{   \mathbf{q}   =  \left(   \boldsymbol{\rho}   -
    \boldsymbol{\sigma}   \right)  \kett{\pi}   +  \frac{k}d
    \mathbf{u} \Big| \Tr \pi = k, \Tr \pi^2 \le k \right\}.
  \end{align*}

  Solutions of $(  \boldsymbol{\rho} - \boldsymbol{\sigma} )
  \kett{ \pi  } = \mathbf{q}  - k  \mathbf{u} / d$  in $\pi$
  exist  if and  only if  $\mathbf{q}  - k  \mathbf{u} /  d$
  belongs   to    the   range   of    $\boldsymbol{\rho}   -
  \boldsymbol{\sigma}$. Solutions are given by
  \begin{align}
    \label{eq:solutions2}
    \kett{\pi} =  ( \boldsymbol{\rho}  - \boldsymbol{\sigma}
    )^+ \left(  \mathbf{q} - \frac{k}d \mathbf{u}  \right) +
    \left( \openone - \Pi \right) \kett{\tau},
  \end{align}
  where  $\Pi :=  ( \boldsymbol{\rho}  - \boldsymbol{\sigma}
  )^+ ( \boldsymbol{\rho} -  \boldsymbol{\sigma} )$, for any
  $\tau \in \mathcal{L} ( \mathbb{C}^d )$.  Notice that $\Pi
  \kett{\openone}   =  0$   since  $(   \boldsymbol{\rho}  -
  \boldsymbol{\sigma}) \kett{\openone} =  k \mathbf{u}/d - k
  \mathbf{u}/d$.     Hence   Eq.~\eqref{eq:solutions2}    is
  equivalent to
  \begin{align*}
    \kett{\pi} = & ( \boldsymbol{\rho} - \boldsymbol{\sigma}
    )^+ \left(  \mathbf{q} - \frac{k}d \mathbf{u}  \right) +
    \lambda \frac{k}d \kett{\openone}\\  & + \left( \openone
    - \frac1d  \kett{\openone} \!\!   \bbra{\openone} -  \Pi
    \right) \kett{\tau},
  \end{align*}
  again for any $\tau \in \mathcal{L} ( \mathbb{C}^d )$.

  The condition $\Tr \pi = k$ immediately implies $\lambda =
  1$.  Moreover, due to the Hilbert-Schmidt orthogonality of
  $(   \boldsymbol{\rho}   -   \boldsymbol{\sigma}   )^+   (
  \mathbf{q}  -   k  \mathbf{u}/d  )$  and   $(  \openone  -
  \kett{\openone}   \!\!    \bbra{\openone}/d    -   \Pi   )
  \kett{\sigma}$, one has that for any $\tau$ such that $\Tr
  \pi^2  \le k$,  the same  condition is  also verified  for
  $\tau =  0$.  Hence,  without loss  of generality  we take
  $\tau = 0$. Thus we have
  \begin{align*}
    \kett{\pi}     =      \left(     \boldsymbol{\rho}     -
    \boldsymbol{\sigma}   \right)^+   \left(  \mathbf{q}   -
    \frac{k}d     \mathbf{u}     \right)     +     \frac{k}d
    \kett{\openone}.
  \end{align*}
  
  Hence,
  \begin{align*}
    &  \Tr  \pi^2  \\  = &  \left(  \mathbf{q}  -  \frac{k}d
    \mathbf{u}   \right)^T    \left(   \boldsymbol{\rho}   -
    \boldsymbol{\sigma}       \right)^{+\dagger}      \left(
    \boldsymbol{\rho} - \boldsymbol{\sigma} \right)^+ \left(
    \mathbf{q} - \frac{k}d \mathbf{u} \right) + \frac{k^2}d.
  \end{align*}
  Thus, condition $\Tr \pi^2 \le k$ becomes
  \begin{align*}
    &  \left( \mathbf{q}  -  \frac{k}d \mathbf{u}  \right)^T
    \left(    \boldsymbol{\rho}     -    \boldsymbol{\sigma}
    \right)^{+\dagger}     \left(    \boldsymbol{\rho}     -
    \boldsymbol{\sigma}   \right)^+   \left(  \mathbf{q}   -
    \frac{k}d \mathbf{u} \right) \\ \le & k - \frac{k^2}d.
  \end{align*}
  Hence the statement follows.
\end{proof}

We are now in  a position to prove Theorem~\ref{thm:image2},
that we rewrite here for convenience.

\begin{thm}
  For   any  $d$-dimensional,   $n$-outcome  informationally
  complete  family $\boldsymbol{\rho}$  of  states, one  has
  that  $   \conv  \cup_{k=0}^d  \mathcal{E}_{d  -   1}^k  (
  \boldsymbol{\rho})$  is  the   maximum  volume  elliptical
  $d$-cone  enclosed in  $\boldsymbol{\rho} (  \mathbb{E} )$
  and     $\conv      \cup_{k=0}^d     \mathcal{E}_1^k     (
  \boldsymbol{\rho})$  is  the   minimum  volume  elliptical
  $d$-cone enclosing $\boldsymbol{\rho} ( \mathbb{E} )$.
\end{thm}

\begin{proof}
  An effect $0 \le \pi \le \openone$ is extremal if and only
  if it  is a projector.   Hence, the set  $\mathbb{E}_d$ of
  effects is the convex hull of projectors, that is
  \begin{align*}
    \mathbb{E}_d  =  \conv   \cup_{k=0}^d  \left\{  \pi  \in
    \mathcal{L} \left( \mathbb{C}^d \right) \Big| \Tr[\pi] =
    k, \pi^2 = \pi \right\}.
  \end{align*}

  The  proof  proceeds  along  the lines  of  the  proof  of
  Theorem~\ref{thm:image}.  First,  due to Sections  6.9 and
  6.11   of  Ref.~\cite{waldron2018introduction},   for  any
  dimension  there   exists  a  finite  scalable   frame  of
  $k$-trace  projectors.   Then,   due  to  Theorem~2.11  of
  Ref.~\cite{CKOPW14},   the    minimum   volume   ellipsoid
  enclosing    $k$-trace    projectors     is    the    ball
  $\mathbb{B}_k^d$.
\end{proof}

\section{Conclusion and outlook}

In this paper we provided an implicit outer approximation of
the image  of any  given quantum  measurement in  any finite
dimension,       thus       generalizing      a       recent
result~\cite{xu2023bounding} by Xu, Schwonnek, and Winter on
the image of Pauli strings.  The outer approximation that we
constructed   is  \emph{minimal}   among   all  such   outer
approximations,  and  \emph{close},  in the  sense  that  it
becomes  the  \emph{maximal  inner} approximation  up  to  a
constant scaling  factor. We also obtained  a similar result
for the dual  problem of implicitizing the image  of the set
of effects through a family  of quantum states.  Finally, we
applied  our approximation  formulas to  characterize, in  a
semi-device independent  way, the  ability to  transform one
quantum measurement into another, or one quantum statistical
model into another.

\section{Acknowledgments}

M.~D.  is grateful to Anna Szymusiak for insightful comments
and  suggestions.  M.~D.    acknowledges  support  from  the
Department  of Computer  Science and  Engineering, Toyohashi
University  of Technology,  from the  International Research
Unit of Quantum Information,  Kyoto University, and from the
JSPS KAKENHI grant number  JP20K03774.  F.  B.  acknowledges
support from MEXT Quantum Leap Flagship Program (MEXT QLEAP)
Grant No.  JPMXS0120319794;  from MEXT-JSPS Grant-in-Aid for
Transformative  Research Areas  (A) “Extreme  Universe”, No.
21H05183;  from JSPS  KAKENHI Grants  No.  20K03746  and No.
23K03230.

\bibliographystyle{unsrturl}
\bibliography{library}
\end{document}